\documentclass[10pt,a4paper]{article}
\usepackage[utf8]{inputenc}
\usepackage{amsmath}
\usepackage{amsfonts}
\usepackage{amssymb}
\usepackage{amsthm}
\usepackage{hyperref}
\usepackage{enumerate}
\usepackage{authblk}

\newtheorem{lemma}{Lemma}

\newcommand{\VTNote}[1]{}
\newcommand{\commentOut}[1]{}

\def\makeItShort{}

\ifx\makeItShort\undefined
    \newcommand{\shortVersion}[1]{}
    \newcommand{\longVersion}[1]{#1}
\else
    \newcommand{\shortVersion}[1]{#1}
    \newcommand{\longVersion}[1]{}
\fi

\begin{document}



\title{Auditing Australian Senate Ballots}

\author[1]{Berj Chilingirian\thanks{Authors are grouped by institution, in alphabetical order, and then listed in alphabetical order within each institution.}}
\author[1]{Zara Perumal}
\author[1]{Ronald L. Rivest} 
\author[2]{\hspace{1cm}Grahame Bowland\thanks{Grahame Bowland is a member of the Australian Greens.  His contribution to this project has consisted entirely of help in implementing the Australian Senate counting rules and facilitating Bayesian audits using his code. The techniques here are non-political.}}
\author[3]{Andrew Conway\thanks{Andrew Conway is a member of the Secular Party.   }}
\author[4]{Philip B. Stark}
\author[5]{\hspace{1cm}Michelle Blom}
\author[5]{Chris Culnane}
\author[5]{Vanessa Teague}

\affil[1]{Computer Science and Artificial Intelligence Laboratory, 
Massachusetts Institute of Technology.  \url{[berjc, zperumal, rivest]@mit.edu}}
\affil[2]{\url{erinaceous.io} \\
\url{grahame@angrygoats.net}}
\affil[3]{Silicon Econometrics Pty. Ltd.\\
\url{andrewsa@greatcactus.org}}
\affil[4]{Department of Statistics, University of California, Berkeley.  \url{stark@stat.berkeley.edu}}  
\affil[5]{Department of Computing and Information Systems, University of Melbourne. \url{[michelle.blom, christopher.culnane, vjteague]@unimelb.edu.au}}

\commentOut{
\author{Berj Chilingirian, 
Zara Perumal, 
Ron Rivest \\
{Massachusetts  Institute of Technology} \\
\url{[berjc, zperumal, rivest]@mit.edu}
}
\author{Grahame Bowland\thanks{Grahame Bowland is a member of the Australian Greens.  His contribution to this project has consisted entirely of help in implementing the Australian Senate counting rules and facilitating Bayesian audits using his code. The techniques here are non-political. } \\ 
\url{erinaceous.io} \\
\url{grahame@angrygoats.net}}
\author{Andrew Conway\thanks{Andrew Conway is a member of the Secular Party.} \\
Silicon Econometrics \\
\url{andrewsa@greatcactus.org}}
\author{
Philip Stark \\
Department of Statistics \\
University of California, Berkeley \\
\url{stark@stat.berkeley.edu}}
\author{Michelle Blom, Chris Culnane, Vanessa Teague \\
Department of Computing and Information Systems \\
University of Melbourne  \\
\url{[michelle.blom, christopher.culnane, vjteague]@unimelb.edu.au}}
}
\maketitle


\begin{abstract}
We explain why the AEC should perform an audit of the paper Senate  ballots against the published preference data files.  We suggest four different post-election audit methods appropriate for Australian Senate elections.  We have developed prototype code for all of them and tested it on preference data from the 2016 election.  
\end{abstract}

\newpage
\tableofcontents
\newpage

\section{Introduction}
A vote in the Australian Senate is a list of handwritten numbers indicating preferences for candidates.  Voters typically list about six preferences, but may list any number from one to more than 200.  Ballots are scanned, digitized and then counted electronically  using the Single Transferable Vote (STV) algorithm \cite{AECSenateCounting}.  

Automating the scanning and counting of Senate votes is a good idea. However, we need to update our notion of  ``scrutiny'' when so much of the process is electronic.  We suggest that, when the preference data file for a state is published, there should be a statistical audit of a random sample of paper ballots.
This should be performed in an open and transparent manner, in front of scrutineers.

Election outcomes must be accompanied by evidence that they accurately reflect the will of the voters. 
At the very least, the system should be \emph{Software Independent}~\cite{rivest2008notion}.
\begin{quote}
A voting system is \emph{software independent} if an undetected change or error in its software cannot cause an undetectable change or error in an election outcome.
\end{quote}

This principle was articulated after security analyses of electronic voting machines in the USA showed that the systems were insecure \cite{feldman2006security, kohno2004analysis, butler2008systemic, CAttbr}. The researchers found opportunities for widespread vote manipulation that could remain hidden, even from well-intentioned electoral officials who did their best to secure the systems.

Followup research in Australia  has shown election software, like any other software, to be prone to errors and security problems \cite{halderman2015new, CountingBug}.  For this reason, evidence of an accurate Senate outcome needs to be derived directly from the paper ballots.  

Legislation around the scrutiny of the count has not kept pace with the technology and processes deployed to perform the count. As a result, the scrutineering has lost a significant portion of its value. With the adoption of a new counting process the scrutineering procedures need to be updated to target different aspects of the system. The current approach might comply with legislation, but it doesn't give scrutineers evidence that the output is correct.

This paper suggests four different techniques for auditing the paper Senate ballots to check the accuracy of the published preference data files.  The techniques vary in their assumptions, the amount of work involved, and the confidence that can be obtained.

These suggestions might be useful in two contexts:
\begin{itemize}
\item if there is a challenge to this year's Senate outcome, 
\item as an AEC investigation of options for future elections.
\end{itemize}

An audit should generate evidence that the election result is accurate, or detect that there has been a problem, in time for it to be corrected.  We hope that these audits become a standard part of Australian election conduct.

\subsection{Q \& A}
\begin{itemize}
\item {\bf Q: Why do post-election audits?} \\ A: to derive confidence in the accuracy of the preference data files, or to find errors in time to correct them.
\item {\bf Q: What can a post-election audit tell you about the election?} \\ A: It can tell you with some confidence that the outcome is correct, or it can tell you that the error rate is high enough to warrant a careful re-examination of all the ballots.
\item {\bf Q: Can the conclusion of the audit be wrong?} \\ A: Yes, with small probability an audit can confirm an outcome that is, in fact, wrong.  It can also raise an alarm about a large error rate, even if the errors do not in fact make the outcome wrong.
\item {\bf Q: Who does post-election audits now?} \\
A: Many US states require by law, and routinely conduct, post-election audits of voter-verified paper votes when the tallies are conducted electronically.  Exact regulations vary---the best examples are the Risk-Limiting Audits \cite{bretschneider2012risk} conducted by California and Colorado.

\item {\bf Q: What is needed to do a post-election audit?} \\
    A: The audit begins with the electronic list of ballots, and (usually) relies on being able to retrieve the paper ballot corresponding to a particular randomly-chosen vote in the file.  There must also be  time and people to retrieve the paper ballots and reconcile them with the preference data file.  A video of random ballot selection is here: \url{https://www.youtube.com/watch?v=sdWL8Unz5kM}.
\item {\bf Q: How long does it take?  How many ballots must be examined?} \\ A: It depends on the audit method, the level of confidence derived, the size of the electoral margin and the number of errors in the sample.  This is described carefully below.    
\item {\bf What is the difference between a statistical post-election audit and a recount?} \\ A: It's not feasible to do manual recounts; a statistical post-election audit would provide a comparable way of assessing the accuracy of the  outcome.
\end{itemize}

\subsection{Our contribution}
This paper describes four suggested approaches to auditing the paper evidence of Australian Senate votes, each described in more detail in Section~\ref{sec:options}.  
\begin{description}
	\item[Section \ref{subsec:ShortBayesianAudits}] Bayesian audits~\cite{rivest2012bayesian},
	\item[Section \ref{subsec:upperBounds}] a ``negative'' audit based on an upper bound on the margin, 
	\item[Section \ref{subsec:fixedSize}] a simple scheme with a fixed sample size,  
	\item[Section \ref{subsec:conditional}] a ``conditional'' risk-limiting audit, which tests one particular alternative election outcome.

\end{description}

We have prototype code available for completing any of the above kinds of audit.  
\longVersion{The further work section contains another suggested scheme for future elections.}
This would be the first time these sort of auditing steps are being applied, and so this year's efforts would be much more ``exploratory'' in character than ``authoritative''.  We hope to be able to perform two or more kinds of audits on the same samples.  However, we do not even know, at the time of writing, whether any audit will happen at all.  

The key objective is to provide evidence that the announced election outcome is right, or, if it is wrong, to find out early enough to correct it by careful inspection of the paper evidence.

\subsection{Where the Senate count depends on trusting software}
This very brief security analysis of the current process is based on documents on the AEC's website  \cite{AECSenateDesign}.
The objective of the system is, in principle, extremely simple:
capture the vote preferences from the ballot papers, and then publish
and tally them. 

The current implementation results in a number of points of trust, in which the integrity of the data is not checked by humans and is dependent on the secure and error-free operation of the software. Whilst internal audit steps are useful, there are many systematic errors and security problems they would not detect.  We list the three most obvious examples below. 

\begin{description}
	\item[Image Scanning] There appears to be no verification that the scanned image is an accurate representation of the paper ballot. As such, a malicious, or buggy, component could alter or reuse a scanned image, which would then be utilised for both the automatic and manual data entry. This would pass all subsequent scrutiny, whilst not being an accurate representation of the paper ballot.  We understand that scrutineers can ask to see the paper ballot, but this seems very unlikely to happen if the image is clear and the preferences match.
	\item[Ballot Data Storage] Whilst a cryptographic signature is produced at the end of the scanning and processing stage, and prior to submission to the counting system, this signature is based on whatever is in the database. There is no verification that the database accurately represents what was produced by the automatic recognition or the manual operator, nor that it was the same thing displayed to scrutineers on the screen. An error, or malicious component, with access to the database could undetectably alter the contents. 
	\item[Signature Checking] Automatic signature generation is a problem in the presence of a misbehaving device. There is no restriction on the device creating signatures on alternative data. Likewise, there appears to be no scrutiny over the data being sent between the scanning process and the counting process, particularly, that the sets of data are equal. There appears to be logging emanating from both services, but no clear description of how such logs will be reconciled and independently scrutinised. 
\end{description}

In summary, there are plenty of opportunities for accidental or deliberate software problems to cause a discrepancy between the preference files and the paper votes.  This is why the paper ballots should be audited when the preference files are published.

\subsection{Background on audits}
The audit process begins with the electronic data file that describes full preferences for all votes in a state.  This file implies a \emph{reported election outcome} $R$, which is a set of winning candidates which we assume to be properly computed from the preferences in the data file.  (Actually we don't have to assume---we can check by rerunning the electronic count.)  Each line in the data file is a \emph{reported vote}---we denote them $r_1,\ldots,r_n$, where $n$ is the total number of voters in the state.  Each reported vote $r_i$ (including blank or informal ones) corresponds to an \emph{actual vote} $a_i$ expressed on paper, which can be retrieved to check whether it matches $r_i$\VTNote{\footnote{What about omissions?}}.   The whole collection of actual votes implies an \emph{actual election outcome} $A$.  We want to know whether $A = R$.

The audit proceeds by retrieving and inspecting a random sample of paper ballots.  A \emph{comparison} audit chooses random votes from the electronic data file and compares each one with its corresponding paper ballot. The auditor records discrepancies between the paper and electronic votes.  A \emph{ballot polling} audit chooses paper ballots at random and records the votes, without using the electronic vote data.  

Although the security of paper ballot processing is important, it's independent of the audit we describe here.  An audit checks whether the electronic result accurately reflects the paper evidence.  Of course if the paper evidence wasn't properly secured, that won't be detected by this process.  Our definition of ``correct'' is ``matching the retained paper votes.''

An election audit is an attempt to test the hypothesis ``That the reported election outcome is incorrect,'' that is, that $R \neq A$.  There are two kinds of wrong answer: an audit may declare that the official election outcome is correct when in fact it is wrong, or it may declare that the official outcome is wrong when in fact it is correct.  The latter problem is easily solved in simpler contexts by never declaring an election outcome wrong, but instead declaring that a full manual recount is required.  The first problem, of mistakenly declaring an election outcome correct when it is not, is the main concern of this paper.

An audit is \emph{Risk-limiting} \cite{lindemanStark12} if it guarantees an upper bound on the probability of mistakenly declaring a wrong outcome correct.  A full manual recount is risk-limiting, but prohibitively expensive in our setting.  None of the audits suggested in this paper is proven to be risk limiting, however all of them provide some way of estimating the rate of errors and hence the likelihood that the announced outcome is wrong.  In some cases, the audit may not say conclusively whether the error rate is large enough to call the election result into question.  In others, we can derive some confidence either that the announced outcome is correct or that a manual inspection of all ballots is warranted.

\subsection{Why auditing the Australian Senate is hard}

Election auditing is well understood for US-style first-past-the-post elections but difficult for complex voting schemes.  The Australian Senate uses the Single Transferable Vote (STV).  
There are many characteristics that make auditing challenging:

\begin{itemize}
	\item {\bf It is hard to compute how many votes it takes to change the outcome.}  Calculating winning margins for STV is NP-hard in general~\cite{Xia2012}, and the parameters of Australian elections (sometimes more than 150 candidates) make exact solutions infeasible in practice.  There are not even efficient methods for reliably computing good bounds. 	
	\item {\bf A full hand count is infeasible,} since there are sometimes millions of votes in one constituency, 
	\item {\bf In practice the margins can sometimes be remarkably small.}  For example, in Western Australia in 2013 a single lost box of ballots was found to be enough to change the election outcome.  In Tasmania in 2016 there were more than 300,000 votes, but the final seat was determined by a difference of 141 votes 
	(meaning errors in the interpretation of 71 ballots might have altered the outcome).	
\end{itemize}

This makes it difficult to use existing post-election auditing methods.

To get an idea of the fiendish complexity of Australian Senate outcomes, consider the case of the last seat allocated to the State of Victoria in 2013.  Ricky Muir from the Australian Motoring Enthusiasts Party won the seat, in a surprise result that ousted sitting Senator 
Helen Kroger of the Liberal party.  In the last elimination round (round 291), Muir had 51,758 more votes than Kroger, and this was generally reported in the media as the amount by which he won.  However, the true margin was less than 3000 (about 0.1\%).  
If Kroger had persuaded 1294 of her voters, and 1301 of Janet Rice (Greens)’s voters, to  vote instead for Joe Zammit (Australian Fishing and Lifestyle Party), this would have prevented Zammit from being excluded in count 224. Muir, deprived of Zammit's preferences, would have been excluded in the next count, and Kroger would have won.   (Our algorithm for searching for these small margins is described in \longVersion{Section~\ref{sec:winnerElimination}.)}
\shortVersion{the full version of this paper.)}

This change could be made by altering 2595 ballots, in each case swapping two preferences, none of them first preferences, all below the line.  First preferences are relatively well scrutinised in pollsite processes before dispatch to the central counting station.  Other preferences are not.   Also \emph{lowering} a particular candidate's preference wouldn't usually be expected to help that candidate (though we are not the first to notice STV's nonmonotonicity).  So the outcome could have been changed by swapping poorly-scrutinised preferences, half of which seemed to disadvantage the candidate they actually helped, in far fewer ballots than generally expected.  

\section{Overview of available options} \label{sec:options}
This section describes four different proposals and compares them according to the degree of confidence derived, the amount of auditing required, and other assumptions they need to make.  We have already implemented prototype software for running Bayesian Audits (Section~\ref{subsec:ShortBayesianAudits}) and computing upper bounds on the winning margin (Section~\ref{subsec:upperBounds}).  We have tested the code on the AEC's full preference data from some states in the 2016 election---results are described briefly below.  

\subsection{Bayesian Audits} \label{subsec:ShortBayesianAudits}
Rivest and Shen's ``Bayesian audit''~\cite{rivest2012bayesian} evaluates the accuracy of an announced
election outcome without needing to know the electoral margin.  It samples from the posterior distribution over
profiles of cast ballots, given a prior and given a sample of the
cast paper ballots (interpreted by hand).  It only looks at a
sample of the cast paper ballots---it does not compare the sampled
paper ballots with an electronic interpretation of them.

An \emph{profile} is a set of ballots.  The auditor doesn't know the
profile of cast (paper) ballots, and so he works with a probability
distribution $p$ over possible such profiles, which summarises everything
the auditor belives about what the profile of cast ballots may be.

The Bayesian audit proceeds in stages.  Successive stages consider
increasingly larger samples of the cast ballots.

Each stage of the Bayesian audit provides an answer to the question
``what is the probability
of various election outcomes (including the announced outcome),
if we were to examine the complete profile of all cast ballots?''

This question is answered by simulating elections on profiles chosen
according to the posterior distribution based on $p$, and
measuring the frequency of each outcome.

 Each audit stage has three phases:
\begin{enumerate}
\item audit some randomly chosen paper ballots
      (that is, obtain their interpretations by a human),
\item update $p$ using Bayes' Rule,
\item sample from the posterior distribution on profiles determined by $p$
  and determine the election outcome for each;
  measure the frequency of different outcomes.
\end{enumerate}

Like any process that uses Bayes' Rule, choosing a prior is a key part
of the initialization.  The suggestion in~\cite{rivest2012bayesian} is
to allow any political partisan to choose the prior that most supports
their political beliefs.  When everyone (who uses Bayes' Rule
properly) is satisfied that the evidence points to the accuracy of the
announced result, the audit can stop.  For example, the auditors could
agree to stop when 95\% of simulated election outcomes match the
reported outcome.

In the Australian Senate case, we assume that there will be only one
apolitical auditing team (though in future candidate-appointed
scrutineers could do the calculations themselves).  Hence we suggest a
prior that is neutral---if the announced outcome is correct, this
probability distribution will be gradually corrected towards it.

An alternative, simpler version amounts to a bootstrap, treating the
population of reported ballots as if it is the (prior) probability
distribution of ballots, and then seeing how often one gets the same
result for samples drawn from that prior.  
This gives an approximate
indication of how much auditing of paper ballots would be necessary, assuming that 
the paper ballots were very similar to the electronic votes.
We have run this version of the audit 
on the Senate outcome from 2016.  
Table~\ref{tab:bootstrapping} shows the number of samples needed in the bootstrapping version, in order to get 95\% of trials to match the official outcome.
Tasmania is the closest, and the only one that's really infeasible: a sample size of about 
250,000 ballots is needed before 95\% of trials produce the official outcome, which is not much better than
a complete re-examination of all ballots.  This is hardly surprising given the closeness of the result. 
Queensland requires 23,000, which is still only a tiny fraction of the total ballots.  
Apart from that, all the other states require only a few thousand samples.

\begin{table}
\begin{tabular}{lll}
{\bf State} & {\bf Number of votes (millions)} & {Audit sample size (thousands)} \\
NSW & 4.4 & 4.6 \\
NT & 0.1 & 1.5 \\
Qld & 2.7 & 23 \\
SA & 1.1 & 3 \\
Tas & 0.34 & 250 \\
Vic & 3.5 & 6 \\
WA  & 1.4 & 9 \\
\end{tabular}
\caption{Sample sizes for 95\% agreement in bootstrap Bayesian Audit.}
\label{tab:bootstrapping}
\end{table}

We suggest a combination of the bootstrapping method with the retrieval of
paper ballots: have a single short partial ballot in
favor of each candidate, combined with an empirical Bayes approach that
specifies that only ballots of the forms already seen in the sample
(or the short singleton ballots) may appear in the posterior distribution.

Although these audits were designed for complex elections, there are
significant challenges to adapting them to the Australian Senate.  
Running the simulations efficiently is challenging when the count itself takes
some time to run.  Answers to these challenges are described in
\shortVersion{the full version of the paper.}\longVersion{Section~\ref{sec:bayes}.}

\subsection{Upper bounds on the margin plus ``negative'' audits}  \label{subsec:upperBounds}
We have implemented some efficient heuristics for searching for ways to change the election outcome by altering only a small number of votes---the code is available at \url{https://github.com/SiliconEconometrics/PublicService}.  The Kroger/Muir margin described in the Introduction is an example.   \longVersion{The algorithm is described in~Section~\ref{sec:winnerElimination}.} 
We can guarantee that the solution we find is genuine, {\it i.e.} a true way to change the outcome with that number of ballots, but we can't guarantee that it is minimal---there might be an even smaller margin that remains unknown.  The algorithm produces a list of alternative outcomes together with an upper bound on the number of votes that need to change to produce them.

If the error rate is demonstrably higher than this upper bound on the margin, then we can be confident it is large enough to change the election result. Of course, it does not follow that the election result is wrong, especially if the errors are random rather than systematic or malicious.  It means that all the paper evidence must be inspected.  

This allows a ``negative audit,'' which can allow us to infer with high confidence that the number of errors is high enough.   

Suppose there are $N$ ballots in all.
Suppose we know that the outcome could be altered by altering no more than $X$ ballots in all, provided those ballots were suitably chosen.  Suppose
we think the true ballot error rate $p$ (ballots with errors divided by total ballots, no matter how many errors each ballot has) is $q$, with $qN \gg X$; that is, we think the error rate is large enough that the outcome could easily be wrong.  Then a modest sample
of size $n$ should let us infer with high confidence that $pN > X$. 

For example, consider the 2016 Tasmanian Senate result, in which the final margin was 71 out of 339,159 votes (a difference of 141 votes).  We can compute the confidence bounds based on a binomial distribution.  A lower 95\% confidence bound for $p$ if we find 3 ballots with errors in a sample of size 2500 is about 0.0003.  That's much greater than the error rate of 
$71/339,159 = 0.00021$ that would be needed to change the outcome.
If we did find errors at about that rate, it would be strong evidence that a full re-examination of all the paper ballots is warranted.  Code for this and other probability computations in this paper is available at \url{https://gist.github.com/pbstark/58653bbc26f269d4588ea7cd5b2e12bf}.

\longVersion{Section~\ref{sec:comparisonAudits} details the statistical analysis and some sharper ways of distinguishing errors that are likely to make a difference to the result.  \VTNote{Actually I'm not sure it does.  Not sure there are any useful sharper distinctions.  Are the stats just binomial bounds again?}}

\subsection{Audits of fixed sample size} \label{subsec:fixedSize}
A much simpler alternative is to
take a fixed sample size of paper ballots (e.g. 0.1\% of the cast ballots), draw that
many ballots at random and examine them all.

This conveniently puts a ``cap'' on the number of randomly-chosen 
paper ballots to be examined, but 
the audit results may provide less certainty than an uncapped
audit would provide.

\subsubsection{Risk-measuring audits}

Assume now that the aim is to try to find confidence that the election outcome is correct.  This audit could quantify the confidence in that assertion, by computing binomial upper confidence bounds on the overall error rate.    The idea is to find the p-value (or confidence level) that the sample you actually have gives you that the outcome is right.

Even an error rate of 0.0002, {\it i.e.}, two ballots with errors per 10,000 ballots, could have changed the electoral result in Tasmania, depending on the exact nature of those errors. The sample size required to show that the error rate is below that threshold---if it is indeed below that threshold---is prohibitively large. 
If we take a sample of 1,000 ballots and we find no errors that affect the 71 margin, the measured risk is the chance of seeing no errors if the true error rate is 0.0002, {\it i.e.,} $(0.0002)^0 * (1-0.0002)^{1000} = 81\%.$
If we took a sample of 2,000, the measured risk would be $(0.0002)^0 * (1-0.0002)^{2000} = 67\%.$

However, this method might be quite informative for other contests.   Manual inspection of a sample of 1,000 ballots could give 99\% confidence that the error rate is below 0.0046 (46 ballots with errors per 10,000 ballots), if the inspection finds no errors at all. If it finds one ballot with an error, there would be 99\% confidence that the error rate is below about 0.0066 (66 ballots with errors per 10,000 ballots).

Similarly, manual inspection of a sample of 500 ballots could give 99\% confidence that the error rate is below 0.0092 (92 ballots with errors per 10,000 ballots), if the inspection finds no errors at all. If it finds one ballot with an error, there would be 99\% confidence that the error rate is below about 0.0132 (132 ballots with errors per 10,000 ballots).

If more errors are found, this gives a way to estimate the error rate.  If it is large, this would give a strong argument for larger audits in the future.

\subsubsection{Fixed-size samples with Bayesian Auditing} 
We can also derive some partial confidence measures from the
given sample.  For example, you could list, for each candidate,
the precentage of the time that candidate was elected across the Bayesian
experiments.  (Each experiment starts with a small urn filled with
the 14000 ballots, plus perhaps some prior ballots, and expands
it out to a full-sized profile of 14M ballots with a polya's urn
method or equivalent.  This is for a nationwide election; for the
senate the full-size profiles are the size of each senate district.)
Depending on the computation time involved, we might run say
100 such experiments.  So, you might have a final output that says:

\begin{tabular}{ll}
Joe Jones    & 99.1 \% \\
Bob Smith    & 96.2 \% \\
Lila Bean      & 82.1 \% \\
$\ldots$    & \\
Rob Meek       & 2.1 \% \\ 
Sandy Slip      & 0.4 \%   \\
Sara Tune       & 0.0 \%   \\
\end{tabular}

Such results are meaningful at a human level, and show
what can be reasonably concluded from the small sample.

This allows us to have a commitment to a given
level of audit effort, rather than a commitment to a given level
of audit assurance, and then give results that say something about
the assurance obtained for that level of effort.


\subsection{Conditional Risk Limiting Audits} \label{subsec:conditional}

Back to the Tasmanian 2016 example again. One way to examine the issue is to consider the particular, most obvious, alternative hypothesis, {\it i.e.} that the correct election result differs only in changing the final tallies of the last two candidates.  If we assume that all the other, earlier, elimination and seating orders are correct, we can conduct a risk-limiting audit that tests only for the one particular alternative hypothesis.  (Of course, it isn't truly risk limiting because it doesn't limit the risks of other hypotheses.) This may be relevant in a legal context in which a challenging candidate 
asserts a particular alternative. This method would provide evidence that the error rate is small enough to preclude that alternative (if indeed it is), without considering other alternatives.

This can be run as a ballot-level comparison audit, in which the electronic ballot record is directly compared with its paper source.  When an error is detected, its impact on the final margin can be quantified  (a computationally infeasible problem when considering all possible alternative outcomes).  
A risk-limiting audit could be based on the Kaplan-Markov  method from~\cite{stark2008conservative}.  It allows the sample to continue to expand if errors are found: that is, it involves sequential testing. 
At 1\% risk limit, the method requires an initial sample size of about (10/margin), where the margin is expressed as a fraction of the total ballots cast. Here, that's about 0.0002.
A risk limit of 5\% would require hand inspection of roughly 16,000 ballots, assuming no errors were found.

\shortVersion{}\longVersion{More details about this method are in Section~\ref{sec:ConditionalRLAs}.}

\subsection{Summary}
These four different audit methods could each be conducted on the same dataset.  We would generate the sample by choosing random elements of the official preference data file, then fetching the corresponding paper ballot.  The Bayesian Audit and the simple capped scheme would then simply treat the paper ballots as the random sample.  The upper-bounds based scheme and the conditional risk limiting audiit would consider the errors relative to what had been reported.

There are important details in exactly how the audit is conducted.  We suggest that the auditors not see the electronic vote before they are asked to digitize the paper---otherwise they are likely to be biased to agree.  However, we also suggest that they are notified in the case of a discrepancy and asked to double-check their result---this should increase the accuracy of the audit itself.  Details of this process are interesting future work.  It is, of course, important that the audit itself should be software independent.

If the rate of error is high then a high level of auditing is required.  With few or no errors, our best estimates of the necessary sample size for each technique applied to the Tasmanian 2016 Senate are:
\begin{itemize}
	\item for Bayesian audits, about 250,000 samples until 95\% of trials match the official outcome,
	\item for ``negative'' audits, a sample that found 3 or more errors out of 2500 ballots would give a 95\% confidence bound on the error rate (being big enough),
	\item a fixed sample size of 500 or 1000, even with no errors, seems unlikely to be large enough to infer anything meaningful for Tasmania 2016, though it may be useful for other contexts,
	\item a conditional risk-limiting audit  would require about 16,000 ballots for a risk limit of 5\%, assuming no errors were found.
\end{itemize}

Most other states would probably be easier to audit as they do not seem to be as close.

\longVersion{
\section{Bayesian Audits and how to use them in the Australian Senate}
\label{sec:bayes}
\VTNote{Currently just copying Ron's email.}

These elections are ranked-choice ballots with many candidates (maybe 100 or
so).  So, it is impossible to explicitly represent all possible ballots in
order to explicitly give a prior probability to each.  

However, these elections do allow partial ballots, where not
all candidates are listed. (Unmentioned ones effectively rank last.)

One purpose of the prior is to make sure that for small profiles you
aren't misled by small-sample statistics.  It doesn't seem fair
to confirm an election result if ballots in favor of one candidate don't
even appear in the sample.   The prior puts a thumb on the scale in
favor of ``fairness''---ensuring that all candidates are equally and
positively represented at the start.  I think this ``fairness towards
candidates'' is more important than ``fairness towards ballot types''.

So, we are proceeding along the path of representing a prior by
having one ballot cast for each candidate as a partial ballot ({\it i.e.},
listing only that candidate).  

It does not, obviously, give equal weight to all possible ballots.
The prior effectively only has support for the ballots actually 
occurring in the sample, and for the singleton ballots.  

(In the end, we are just effectively computing various ``noisy''
or ``fuzzy'' versions of the drawn sample, and seeing how this
affects the election outcome computation.)

The other approach that comes to mind would be to somehow generate
new ballot types de novo as the audit proceeds, and add them to
the urn.  This is not an entirely unreasonable thought, and there
is stuff in the literature that seems quite similar to this (e.g. the
literature on stick-breaking or Chinese restaurants in the Bayesian
literature).

However, there is another reason why our adopted procedure 
of having one prior ballot per candidate is attractive: efficiency.
With our procedure, the number of distinct ballot types does
not increase during the audit.  The ``Polya's Urn process'' or
equivalent only changes the frequency of each existing ballot
type---it doesn't add new types.

This means that the ``state of the urn'' is always fully representable
as a vector of length $t$, where t is the number of distinct ballot
types in the sample and prior.  (If the sample has size $s$, then $t$
is at most $s+m$ where $m$ is the number of candidates.)

In particular, this means that we can employ the gamma-variate
trick (see the Rivest/Shen paper) to do the audit, replacing the Polya's Urn
computation (which takes time $n$, where $n$ is the number of cast
ballots) with a computation that takes only time $t$ (computing
one gamma-variate for each ballot type, in order to obtain the
final count for the number of ballots of that type).  This improves
running time by a factor of $n/t$, which can be very large (e.g.
$n=10^6$ and $t=10^3$ gives a 1000-fold speedup.)

} 

\longVersion{
\section{Heuristics for upper bounds on STV margins}  \label{sec:winnerElimination}
\VTNote{Add a discussion on Michelle's method for finding upper bounds by looking only at the last few (20 or so) rounds.}

This section explains a way of searching for small manipulations in STV elections.  The code is available at 
\url{https://github.com/SiliconEconometrics/PublicService}.

The last subsection suggests how this could be included in an audit.  The result is not really a risk-limiting audit, but rather an audit that would be risk limiting if certain assumptions are true.  The assumptions will be precise but computationally intractable to check.   Alternatively, the upper bounds could be used in a ``negative'' audit, as described in the introduction, to provide confidence that the error rate \emph{is} large enough to matter.

The idea generalises Cary's ``winner-elimination upper bound'' to STV.  The high-level idea is simple: for each candidate $w$ who won \emph{without gaining a quota}, at each step in which some other candidate $e$ is eliminated, we calculate how many votes would have to be moved from $w$ (or some other candidates with large tallies) to $e$ (and some other candidates with small tallies) to get $W$ eliminated instead.   Clearly this changes the outcome, though we don't immediately know how.
This gives us an upper bound on the size of the manipulation necessary to change the result.  

There may, of course, be smaller manipulations that are not discovered by this method.  
When a solution is found, there may be many different ways to achieve it, corresponding to several different election outcomes.  

There are two slightly different variants, each corresponding to different assumptions about the source of error.  In the first variant, we allow any sort of change to ballots.  In the second, we assume that first-preferences are recorded accurately (in the polling place, under scrutiny) and that the only successful manipulation is one that leaves the first preference tallies unchanged.  In both cases there are a number of different variants and details that matter.  Although the idea generalises to any form of STV, we have implemented it for two particularly interesting ones: the idealised version of STV we used in our computation of margins \cite{blom2015efficient}, and a precise implementation of the Australian Senate counting rules.

The Australian Senate rules have a number of complicating issues, including complex rules for multiple eliminations.  For this reason, we describe the idealised version of the algorithms first, then describe how we deal with some of the complexities of the Senate rules in a later section.

\subsection{Theoretical/ideal Preliminaries}
We have an announced set $W$ of winners.  $W = \{w_1, w_2, \ldots, w_s \}$ where $s$ is the number of seats.  There are two different ways of winning: $w$ may get a quota, or $w$ may be one of only $k$ candidates to remain uneliminated at the end, when there are only $k$ seats remaining to be allocated.  Sometimes these two conditions might be met simultaneously.  Let $L \subseteq W$ be the set of candidates who win by being left uneliminated at the end.  At each step, for each winner $w \in L$, we want to shift enough votes from $w$ to get $w$ eliminated at (or before) that point.  Then we want to rerun the election and see who wins instead.

Two important things to check:
\begin{itemize}
	\item that the prior elimination rounds are not affected or, if they are, that $w$ gets eliminated, and
	\item that if we shift votes that have been used to elect a candidate, the manipulation size is counted in terms of number of paper ballots, not the sum of the weights as a result of transfers.
\end{itemize}

Let $t_i(c)$ be the tally of candidate $c$ at round $i$.

\subsubsection{Variant 1: when first preferences are allowed to change}
We try to remove $m$ votes from $w$ and share that among the low-tally candidates until all other tallies are higher than $w$'s.

For each round $i$, if candidate $e$ is eliminated at round $i$,

For each $w \in W$, set 
\begin{equation} m_i(w) = \min \{m :  \forall \text{continuing candidates }  c \neq w, 
	t_i(w) - m < t_i(c) + \delta_c  \text{ where } \Sigma_c \delta_c = m   \}  
\label{eqn:margin} \end{equation}

\begin{lemma}
If $w$ has $m_i(w)$ first-preference votes, then this will produce a valid solution.  
\end{lemma}
\begin{proof}
If the shift gets $w$ eliminated before round $i$, then that's a valid solution.  

So suppose that $w$ is still standing at round $i$.  Adding extra votes to other candidates can't have affected the elimination order before now, because it could only have increased the tallies of uneliminated candidates.  (Detail: we need to argue that we also haven't accidentally given someone a quota who wouldn't otherwise have had one.  This is true because $w$ must have less than a quota at $i$ (or it's not continuing), and the recipients have received only enough votes to bring them up to $t_i(w)$ at most.)

Also note all the transfer values are 1, so the total number of ballots is equal to the total weight of the votes.
\end{proof}

If we allow non-first preferences to change too, we need to check that the manipulation doesn't affect the elimination order before $i$, or that if it does then $w$ gets eliminated anyway.  A very similar argument applies: simply take the preference that lists $w$ and substitute a preference for $c$ instead (removing any subsequent mention of $c$).  This will affect only tallies too high to have been eliminated by round $i$ and too small to have a quota.  However, we now need to be more careful about accounting for manipulations on ballots with a transfer value less than 1.  The minimum manipulation must be counted in terms of the number of ballots, not the total weight of them.  We explore this idea more precisely below.

\subsubsection{Variant 2: when first preferences must stay constant}
The next section assumes that there is some other, independent and fixed, list of first preferences, and that any manipulation must therefore leave them unchanged.  The idea is the same, but instead of removing (first) preferences from $w$, we remove (non-first) preferences for $w$ that are currently sitting on $w$'s pile.  These preferences are shifted across to (other) low-tally candidates, until enough have been shifted to get $w$ eliminated.  It's possible that a solution might exist in Variant~1, but none in Variant~2.

If there aren't enough available preferences on $w$'s pile, we can take some preferences from candidate higher than $w$.  This could sometimes affect the earlier elimination order, so the resulting solution must be re-run through the count to check that it is valid.

Now because we're not dealing with first preferences, we need to be careful to account for whole ballot papers even when the transfer values are less than one.  We need to minimise the total number of ballots shifted, given the same restrictions on tallies expressed in Equation~\ref{eqn:margin}.

\subsection{Practicalities and optimisations for the Senate}
The above section gives simple computations for computing the bounds directly.  However, the true counting algorithm for the Australian Senate is complex and fiddly.  For example, these modifications may affect rounding and the rules for multiple eliminations in ways that are hard to test for.  Other senate practicalities:
\begin{itemize}
\item  ATL votes are harder to manipulate (before 2016, impossible. after 2016, tricky)
\item  Need to use a different technique (stuff from other people) to get anything in 2013. Binary search is not used here.
\item  Tie resolution is messy.  We solve this by just overcompensating a little, to make sure there are no ties.
\item  Transfer values make which votes  to use messy (e.g a 0.4TV from $w$ vs a 1.0TV from some high scorer). We arbitrarily (and potentially inefficiently) assume that all votes from $w$ are used in preference to other people’s, when first-preference changes are allowed, but otherwise order by transfer value.    
\item Rounding is very hard to account perfectly for, and may produce small errors.
\end{itemize}

We solve this problem by simply doing a binary search for $m$, rather than computing $m$ directly, and recomputing the election outcome on the modified data to check that it is correct. 

\subsubsection{Further optimisations}
In order to eliminate $w$, it might not be necessary to eliminate them at the round where they are numerically closest to elimination.  In other words, the description above might be correct, but we might be able to achieve the same effect with fewer manipulations, if we let some of the low-tally candidates get eliminated before $w$.  This suggests the following expanded search:

Once finding the best solution using the techniques described above, take each low-tally candidate $l$ and let $r(l)$ be the votes that $l$ has received.  Do a binary search adding between 0 and $r(l)$ votes to $l$, to see whether $w$ does indeed get eliminated (in a later round) when $l$ receives fewer than $r(l)$ votes.  Reset $l$'s extra votes to the minimum $r'(l)$ for which $w$ still gets eliminated.  Iterate through all low-tally candidates, twice.  In practice we find this significantly reduces some of the margins we found.

This method discovered the Muir/Kroger example described in the Introduction.

\VTNote{I reckon we need a table like Michelle's table, listing apparent (last round) margin, margin by changing first pref's, margin without changing first pref's, and optimized margin, for a few key races from 2013 and 2016.}<++>

} 

\longVersion{
\section{Ballot-level comparison audits against a particular alternative: statistics and counting details} \label{sec:conditionalRLAs}

We do not know how to conduct risk-limiting audits on Australian Senate elections.  However, sometimes a dispute is oriented around a single, specified, alternative election outcome.  The most obvious of these is to switch the winner among the last two candidates in the running for the final seat.  This extends easily to the case when $k+1$ candidates are still present when $k$ seats are available---the lowest loses, and the other $k$ win seats. 


A \emph{ballot-level comparison audit} selects a random sample of electronic ballots from the file, then compares each individual paper ballot with its electronic record.  Discrepancies are carefully counted---we want to understand whether there are enough discrepancies to suggest that the paper ballots produce a different outcome from the official (electronic) one.   It is \emph{risk-limiting} if the probability of incorrectly accepting a wrong outcome is bounded.

We do not know how to make such an audit truly risk-limiting for Australian Senate races, because it is infeasible to compute the effect of discrepancies, particuarly when they are combined (remember the Muir/Kroger example).  However, if we are only testing against one particular hypothesis, and that hypothesis doesn't change any of the early elimination and seating orders, we can perform a ``conditional'' risk limiting audit.  It is conditional in the sense that it bounds the probability of wrongly accepting the official outcome when the truth is alternative hypothesis $H$, though it guarantees nothing about still-other hypotheses.

In particular, we design a ballot-level comparison RLA around falsifying the hypothesis, ``That the outcome is wrong because of a change that switched the last two candidates \emph{but didn't affect the order of eliminations or seatings in prior rounds.}

This is relevant, for example, in the 2016 Tasmanian example.  We don't know whether the 141-vote final margin was minimal, but this kind of audit could test just that margin.

Suppose the last-round difference is $M$, between candidates $w$ (the winner) and $l$ (the loser).  Consider how to examine any discrepancy (between data file and paper) and assess exactly how much it contributed to that final margin, assuming the apparent elimination and seating order.  

For example, if a vote for $l$ appears, or a vote for $w$ disappears, after a candidate who's seated with 2 * quota (and no other quota-getters), then we can say that that error would take $1/2M$ off the margin. 

We list the discrepancies according to what appears (or disappears) on the paper ballot, compared to what's in the data file.  The obvious ones are those that mention $l$ or $w$ explicitly:

\begin{tabular}{ll}
{\bf Discrepancy} & {\bf Effect} \\

$l$ disappears & none \footnote{Don't count discrepancies that enhance the likelihood of the official solution.}\\
$l$ appears below $w$ & none \\
$l$ appears with no seated candidates above & 1 \\
$l$ appears below a series of candidates seated & \\
 with a quota, the last one\footnote{Under Australian Senate rules, only the last transfer value matters.  With a weighted transfer method, this would have to consider the product of all effects from successive surplus transfers.} with surplus $s$ & $s/(1+s)$ \\
$w$ appears & none \\
$w$ disappears below $l$ & none \\
$w$ disappears with no seated candidates above & 1 \\
$w$ disappears below a series of candidates seated & \\
 with a quota, the last one\footnote{Under Australian Senate rules, only the last transfer value matters.  With a weighted transfer method, this would have to consider the product of all effects from successive surplus transfers.} with surplus $s$ & $s/(1+s)$
\end{tabular}

More subtle, but important, discrepancies are those that affect the transfer values themselves.  Suppose that a candidate $c$ who transferred half their preferences to $l$ when seated actually has a far larger surplus than the electronic records indicate.  If we disover evidence of this, we need to count it. The appearance or disappearance of a votes for a seated candidates can, in this way, affect the final tallies by at most $1/q$.  \VTNote{**Check.} Hence we count all such errors as $1/qM$.

\VTNote{
Then all we need is a way to tally up errors which all count somewhere between 0 and $1/qM$ towards the same margin.  Can I adapt MACRO from ``Conservative statistical post-election audits''? Or is there a reason that multiple-voting doesn't necessarily correspond to different magnitudes of shift?}

This then gives us an RLA with an explicit (probably widely accepted) assumption.  This would say something specific, though with the big assumption of having eliminated only the one most plausible-looking kind of wrong result.  If the other methods based on Bayes and Winner-Elimination heuristics and a fixed sample were inconclusive, this might be a reasonably convincing argument that the announced result is right.  

For small margins, the sample size would still be large.  
For example, consider the sample size needed to show with high confidence that there are not 141 ``important'' errors in the Tasmania 2016 votes.
We'd have to find zero errors that mattered in a random sample of about 7,500 ballots to conclude that the population error rate was less than 0.0004, at 95\% confidence.
}

\section{Implementation Summary}
All the tools necessary for conducting a Bayesian audit of Australian Senate votes are available as a Python package at \url{https://pypi.python.org/pypi/aus-senate-audit}, with code and instructions at \url{https://github.com/berjc/aus-senate-audit}.

Code for searching for small successful manipulations is at \url{https://github.com/SiliconEconometrics/PublicService}.

Code for computing relevant statistical bounds is at \url{https://gist.github.com/pbstark/58653bbc26f269d4588ea7cd5b2e12bf}.

\section{Conclusion}
Elections must come with evidence that the results are correct.  This work contributes some techniques for producing such evidence for the partly-automated Australian Senate count.

All of the audits discussed here can be conducted immediately, using code already available or specifically produced as a prototype for this project.

\subsection{Future Work}
In the future we could expand the precision with which we record errors and make inferences about their implications.
We are also pursuing an easier user interface for administering the audit.

\longVersion{

On Thu, Aug 4, 2016 at 3:04 AM, Damjan Vukcevic <damjan@vukcevic.net> wrote:

    Hi everyone,

    I'm back from my holiday.  I see that you have discussed this topic quite extensively.  Vanessa summarised my proposal earlier, while I was away.  Hopefully you understood the main ideas.  In any case, she has asked me to reiterate it in a bit more detail.  I've done this below.

    While I haven't read through the whole email discussion in detail (so you might have mentioned this already somewhere), I think the key difference between a 'black box' type of approach and one like mine is that we get estimates of the error rates, and don't counfound this with any uncertainties due to the vote being a close contest.  Thus, we have the ability to make some conclusions about the general reliability of the counting process, independently of any specific election outcome.  I think that is important in being able to make a case for the utility of these methods, especially in the current Senate election context where we have many close results.

    Please note that when I initially spoke with Vanessa, I wasn't familiar with any of the existing literature.  So I guess it's no surprise that there are similarities between my ideas and some of your work.  If anything, that's reassuring!  (FYI, I haven't yet read the Bayesian auditing paper, so I'm not sure how much what I've written below overlaps with it.)

    At a high level, the idea had two steps:

\begin{enumerate}
\item Use a sample of the ballots to determine estimates of errors in the whole ballot counting process (i.e. including all of the steps, such as scanning, character recognition, etc.).  \label{sample}

\item Use these estimates together with the final published counts to determine the certainty of the final set of elected candidates, by simulating the effect of the errors across the whole set of counts.  \label{useEstimates}
\end{enumerate}

    Since the aim is to compare the paper ballots to the final digital file of preferences, and there is a possibility that some of these won't match perfectly (indeed, we wish to detect how often this is the case), for step 1 it would be necessary to take random samples from both sources.  That means that, at a high level, we'll end up cataloguing three types of errors:

\begin{enumerate}[(a)]
    \item Paper ballots with no matching digital record
    \item Digital records with no matching paper ballots
    \item Paper ballots that do not match their digital record \label{paperNoMatch}
\end{enumerate}

    Errors of type~\ref{paperNoMatch}  can then be further classified by the nature of the mismatch, if practical or desired.  Given the large number of possible votes, it would be impractical to estimate a full mismatch matrix.  However, we can try simple things like seeing how often a first preference vote for candidate X is mistranscribed as a first preference vote for candidate Y.

    Step~\ref{useEstimates}  would be naturally cast in a Bayesian framework, where we wish to calculate a posterior distribution on the set of elected candidates.  We would do so by simulating errors in a way that is consistent with the estimates from step~\ref{sample}.  (There are clearly many ways one might set this up, I haven't thought about exactly what the best way is.  It is related to the question of how to summarise the errors of type~\ref{paperNoMatch}  above.)  We can replicate this many times and thus get a probability distribution.

    On its own, Step~\ref{sample} would be informative about the nature and reliability of the whole counting process, which is the main issue of interest.  Step~\ref{useEstimates} takes it further and translates this to the results of the current election, thus quantifying the impact of the errors (and our knowledge of them) to the election results.  The latter is easier to interpret, but it is answering a slightly different (but nonetheless relevant) question.

The challenge is classifying errors into the
error types of a model, so that you can replay the model
against a sample of ballots.   This only gets applied
``one-way'', though --- you want to apply the error model
learned from a sample of paper ballots to all of the electronic
records corresponding to paper ballots that haven't been
examined.
} 
\bibliographystyle{alpha}
\bibliography{e-vote}
\end{document}